%% file: main.tex
\documentclass[12pt]{article}
\usepackage[dvipsnames,svgnames,x11names,hyperref]{xcolor}
\usepackage{amsmath}
\usepackage{amsthm}
\usepackage[T1]{fontenc}
\usepackage{fullpage}
\usepackage{thmtools}
\usepackage[colorlinks]{hyperref}
\hypersetup{
  linkcolor=[rgb]{0,0,0.4},
  citecolor=[rgb]{0, 0.4, 0},
  urlcolor=[rgb]{0.6, 0, 0}
}
\usepackage{amssymb,amsfonts,amsmath,amsthm,amscd,dsfont,mathrsfs}
\usepackage{complexity}
\usepackage{tikz}
\usetikzlibrary{decorations.pathreplacing,backgrounds}
\usepackage{multirow}
\usepackage{mathpazo}
\usepackage{braket}
\usepackage{algorithm}
\usepackage{algpseudocode}

\input{thmmacros}
\input{customurlbst/bibmacros}

\newcommand*\samethanks[1][\value{footnote}]{\footnotemark[#1]}

\newcommand{\Max}{\text{M}}
\renewcommand{\C}{\mathbb{C}}

\renewcommand{\R}{\mathbb{R}}
\renewcommand{\set}[1]{\left\{ #1 \right\}}
\newcommand{\abs}[1]{\left| #1 \right|}
\newcommand{\ord}[2]{\text{ord}_{#1}{#2}}
\newcommand{\norm}[1]{\left\Vert {#1} \right\Vert}
\newcommand{\size}[1]{\text{size}\left( #1 \right)}

\newcommand{\Z}{\mathbb{Z}}
\newcommand{\N}{\mathbb{N}}
\newcommand{\Q}{\mathbb{Q}}
\newcommand{\F}{\mathbb{F}}

\newcommand{\lc}{\text{LC}}
\newcommand{\tc}{\text{TC}}
\newcommand{\rev}{\text{rev}}

\newcommand{\ceil}[1]{\lceil #1 \rceil}
\newcommand{\floor}[1]{\left\lfloor #1 \right\rfloor}
\newcommand{\equivrelation}[3]{#1 \; \equiv_{#3} \; #2}

\algnewcommand\algorithmicinput{\textbf{Input:}}
\algnewcommand\algorithmicoutput{\textbf{Output:}}
\algnewcommand\Input{\item[\algorithmicinput]}
\algnewcommand\Output{\item[\algorithmicoutput]}

\newcommand{\ignore}[1]{}

\title{A New Bound on Cofactors of Sparse Polynomials}
\author{
	Ido Nahshon
	\thanks{
		Blavatnik School of Computer Science, Tel Aviv University, Tel Aviv, Israel, \texttt{\{idonahshon,shpilka\}@tauex.tau.ac.il}.
		This research was co-funded by the European Union by the European Union (ERC, EACTP, 101142020), the Israel Science Foundation (grant number 514/20) and  the Len Blavatnik and the Blavatnik Family Foundation. Views and opinions expressed are however those of the author(s) only and do not necessarily reflect those of the European Union or the European Research Council Executive Agency. Neither the European Union nor the granting authority can be held responsible for them.
	}
	\and
	Amir Shpilka\samethanks[1]
}
\date{}

\begin{document}

	\maketitle

\abstract{

We prove that for polynomials \( f, g, h \in \mathbb{Z}[x] \) satisfying \( f = gh \) and \( f(0) \neq 0 \), the \(\ell_2\)-norm of the cofactor \( h \) is bounded by  
\[
\|h\|_2 \leq \|f\|_1 \cdot\left( \widetilde{O}\left(\|g\|_0^3 \frac{\deg(f)^2}{\sqrt{\deg(g)}}\right)\right)^{\|g\|_0 - 1},
\]
where \(\|g\|_0\) is the number of nonzero coefficients of $g$ (its sparsity). We also obtain similar results for polynomials over $\C$.\footnote{These bounds improve upon the conference version of this paper \cite{NahshonShpilka24}.} 

This result significantly improves upon previously known exponential bounds (in \(\deg(f)\)) for general polynomials. It further implies that, under exact division, the polynomial division algorithm runs in quasi-linear time with respect to the input size and the number of terms in the quotient  \( h \). This resolves a long-standing open problem concerning the exact divisibility of sparse polynomials.

In particular, our result demonstrates a quadratic separation between the runtime (and representation size) of exact and non-exact divisibility by sparse polynomials. Notably, prior to our work, it was not even known whether the representation size of the quotient polynomial could be bounded by a sub-quadratic function of its number of terms, specifically of \(\deg(f)\).

\ignore{	
	We also study the problem of bounding the number of terms of $f/g$ in some special cases.
	When $f, g \in \Z[x]$ and $g$ is a cyclotomic-free
        (i.e., it has no cyclotomic factors) trinomial,
	    we prove that $\norm{f/g}_0\leq O(\norm{f}_0 \size{f}^2 \cdot \log^6{\deg{g}})$.
    When $g$ is a cyclotomic free binomial,
	    we prove that the sparsity is at most
	    $O(\norm{f}_0 ( \log{\norm{f}_0} + \log{\norm{f}_{\infty}}))$.
    Both upper bounds are polynomial in the input-size.
	Without assuming cyclotomic-freeness, there are simple examples 
	    that show that both sparsities can be exponentially large. 
        Leveraging these results, we provide a polynomial-time algorithm for
	    deciding whether a cyclotomic-free trinomial
	    divides a sparse polynomial over the integers.
	In addition, we show that testing divisibility by a binomial can be done in $\widetilde{O}(n^2)$
	    time over $\Z$ and over finite fields.
	
	As our last result, we present a polynomial time algorithm
	    for testing divisibility
	    by pentanomials over small finite fields
	    when $\deg{f} = \widetilde{O}(\deg{g})$.
	
	All our algorithms are deterministic.
}
}

\thispagestyle{empty}
	\newpage
	\tableofcontents
\thispagestyle{empty}
	\newpage
\pagenumbering{arabic}

\section{Introduction}
    \label{sec:intro}

	Let $f,g,h\in \C[x]$ be polynomials  such that $f=gh$. 
	Classical results of Gel'fond, Mahler and Mignotte provide upper bounds on the height of $h$ ($\ell_\infty$ norm of its vector of coefficients) given information on  the coefficients of $g$ and $f$ \cite{gelfond60,mahler:1962,mignotte:1974,mignotte:1988}. 
	These results were used by Gel'fond in the study of transcendence and  are widely used in computer algebra to prove upper bounds on the running time of algorithms, such as in the polynomial factorization algorithm of Lenstra, Lenstra, and Lov\'asz \cite{lenstra-lovasz:1982}.  

	The first of these bounds was proved by Gelf'ond  \cite{gelfond60} (see also \cite{mahler:1962} and \cite[Lemma 1.6.11]{bombieri:2006}).          
           \begin{theorem}[ Gel'fond's Lemma]\label{thm:gelfond}
           	Let $f, g, h \in \C[x]$ such that $f = gh$.
           	Then,
           	 \begin{equation}
           		\label{eq:gelfonds-lemma}
           		\norm{h}_{\infty}
           		\leq 2^{\deg{f}}\frac{\norm{f}_{\infty}}{\norm{g}_{\infty}}.
           	\end{equation}
           \end{theorem}
	Improved bounds were provided by Mahler \cite{mahler:1962} and consequently by Mignotte \cite{mignotte:1974}.
            \begin{theorem}[Mignotte's Bound {\cite[Theorem 2]{mignotte:1974}}]\label{thm:Mignotte}
                Let $f, g, h \in \Z[x]$ such that $f = gh$.
                Then,
		  \begin{equation}
                \label{eq:mignotte-bound}
                \norm{h}_1 \leq 2^{\deg{h}}\norm{f}_2
            \end{equation}
            \end{theorem}
            Discussions about this classical inequality
            can be found in \cite[Section 2]{mignotte:1988}
            and \cite[Chapter 4]{mignotte:2012}.
            This bound is also known to be quite tight;
            for any $d \in \N$ there exists $h \in \Z[x]$ of degree $d$
            and $f \in \Z[x]$ such that
            \begin{equation*}
                \norm{h}_1 \geq \Omega \left(\frac{2^{d}}{d^2\log{d}} \right) \norm{f}_2
                .
            \end{equation*}
 
	When $\|h\|_0$ is guaranteed to be   small relative to its degree, tighter bounds can be obtained. For instance, Giorgi, Grenet and Perret du Cray established the following bound by a straightforward induction on the polynomial division algorithm \cite{giorgi-grenet-cray:2021}.
 	
\begin{lemma}[Lemma 2.12 in \cite{giorgi-grenet-cray:2021}]
	Let $f,g,h \in \Z[x $] satisfy $f=gh$.  Then

 \begin{equation}
 	\label{eq:euclidean-division-bound}
 	\norm{h}_{\infty}
 	\leq \norm{f}_{\infty} (\norm{g}_{\infty} + 1)^{\frac{\norm{h}_0+1}{2}}.
 \end{equation}
 \end{lemma}
	While the bound given in \eqref{eq:mignotte-bound} is close to being tight, in this paper we ask whether one can get an improvement  if the factor $g$ is sparse. For example, Khovanski\u{i}'s theory asserts that when studying real roots of polynomial equations, results analogous to B\'ezout's theorem can be obtained, where the number of monomials in the equations takes the place of the degree in B\'ezout's theorem \cite{khovanskii}.

Besides being a self-motivated question, bounding the norm of cofactors of sparse polynomials have important applications in computer algebra which we discuss next.

\subsection{The Sparse Representation}
    The sparse representation of polynomials
        is widely used in computer algebra as it is the most natural
        and succinct representation when polynomials are 
        far from being ``dense'', i.e., when most of their coefficients are zero,
        or, equivalently, when they are supported on a small number of monomials
        \cite{modern-computer-algebra:2013,computer-algebra-systems:1993}.
    In this representation, a polynomial is described
        by a list of coefficient and degree pairs
        that correspond to its terms with non-zero coefficients.
    For example, a polynomial $f \in R[x]$ over the ring $R$
        with $t$ non-zero coefficients, $f = \sum_{i=1}^{t}{c_i x^{e_i}}$, 
        is represented as $\set{(c_i,e_i) \mid i = 1 \ldots t}$.
    Particularly interesting cases are when $R = \Z$ or $R$ is a finite field.
    When working over the integers,
        the size of the sparse representation
        of a polynomial $f \in \Z[x]$, denoted $\size{f}$, is
    \begin{equation*}
        \size{f} = \norm{f}_0(\log_2{\norm{f}_{\infty}} + \log_2{\deg{f}})
        ,
    \end{equation*}
    where $\norm{f}_0$ and  $\norm{f}_\infty$ are the number of non-zero coefficients in $f$,
    and the maximal coefficient of $f$ (in absolute value), respectively.\footnote{
        As we think of $\norm{f}_0$ as being considerably smaller than
        $\log_2{\norm{f}_{\infty}}$ and $\log_2{\deg{f}}$,
        we do not pay attention to the case where bit lengths of coefficients vary.
    }
\ignore{    Similarly, over a finite field $\F$ the representation size is
    \begin{equation*}
        \size{f} = \norm{f}_0(\log_2{\abs{\F}} + \log_2{\deg{f}}).
    \end{equation*}
}
    This makes the sparse representation especially useful for applications involving polynomials with high degrees where most of the coefficients are zero. In particular, it is used by computer algebra systems and
        libraries such as Maple \cite{monagan2013poly},
        Mathematica \cite{wolfram1999mathematica}, Sage \cite{stein2008sage},
        and Singular \cite{schonemann2003singular}.
        
    In the sparse representation, the degree could be exponentially larger
        than the representation size.
    Algorithms that run in polynomial time in terms of the sparse-representation size
        (i.e. number of terms, their bit complexity and the logarithm of the degree)
        are called super-sparse or lacunary polynomial algorithms.
        
 \sloppy   Although the basic arithmetic of sparse polynomials is fairly straightforward
    \cite{johnson:1974}, the sparse representation poses many challenges
        in designing efficient algorithms, and many seemingly simple problems
        are known to be NP-hard.
	For example, Plaisted proved that determining whether
        $x^n-1$ divides a product of sparse polynomials 
	    (such that the size of the input is $\poly(\log n)$)
        is NP-complete \cite{plaisted:1977,plaisted:1984}.
    Other NP-hard problems concerning sparse polynomials include:
        deciding whether there exists a common divisor of two polynomials
        and testing whether a polynomial has a cyclotomic factor \cite{plaisted:1977}.
    The surveys of Davenport and Carette
	    \cite{davenport-carette:2009} and of Roche \cite{roche:2018} give a good picture
	    of known results and open questions concerning sparse polynomials.    

    One of the long-standing open problems concerns exact divisibility of sparse  polynomials.
    There are two versions to this problem. The decision problem asks, given two sparse polynomials $f$ and $g$, to decide whether $g$ divides $f$. The search problem, or the exact division problem, asks to compute the quotient polynomial $f/g$ when it is known that $g$ divides $f$. 
    
    The two problems are closely related but also differ in what we can hope to achieve.
    In the decision problem, the output is a single bit, allowing us to analyze the complexity in terms of the input size.
    Conversely, in the search (exact division) problem, the output size depends on the number of monomials in $f/g$, which may require an exponential number of bits (compared to the size of the input)  to represent.
    
    
     There are two quantities that affect the  complexity of the exact division problem. The first is the number of terms of the quotient $f/g$. It is easy to find examples in which  $f/g$ has exponentially many monomials. The second quantity is the height. 
Prior to this work  the only bounds on the height were exponential in $\deg{f}$, as in Equations~\eqref{eq:gelfonds-lemma},\eqref{eq:mignotte-bound}  and\eqref{eq:euclidean-division-bound} (see also \cite{mignotte:2012}). 
In particular, Davenport and Carette note that even in the case where $g$ exactly divides $f$, we need to perform $\Omega(\|f/g\|_0\cdot \|g\|_0)$ operations on coefficients,
        and the number of bits required to represent a coefficient may be as large as $ \deg f$ \cite[Section III.B]{davenport-carette:2009}.
        Therefore, in term of the input size and $\|f/g\|_0$ the only provable bounds on the running time of the exact division algorithm were of the form 
        \begin{equation}\label{eq:DC}
        	\Omega(\|f/g\|_0\cdot\|g\|_0\cdot  \deg f).
        \end{equation}

\ignore{            
    \subsection{Divisibility Testing}

	Recall that in the divisibility testing problem the goal is to determine
        whether $g$ divides $f$ without having to compute the quotient and remainder.
        Thus, we seek a polynomial-time algorithm in the size of the input.
    It is easy to see that when $\deg{g}$ is small (or, similarly, $\deg{f}-\deg{g}$ is small),
        or when the factor polynomial is sparse, the problem becomes easy.

    However, in the general case, it is only known that,
        assuming the Extended Riemann Hypothesis (ERH), the divisibility problem is in coNP
        \cite{grigoriev-karpinski-odlyzko:1992}.
    This led Davenport and Carette to pose the following challenge:

    \begin{openproblem}[Challenge 3 of \cite{davenport-carette:2009}]
        \label{op:davenport-carette}
        Find a class of instances for which the question
        "does g divide f?" is NP-complete,
        or find an algorithm for the divisibility of polynomials
        which is polynomial-time.
        Failing this - find an algorithm for the divisibility of
        cyclotomic-free polynomials which is polynomial-time.
    \end{openproblem}

In fact, even the  next problem, that was stated in 
       \cite{davenport-carette:2009,roche:2018,giorgi2022sparse}
}

Following \cite{davenport-carette:2009}, Roche posed the next problem \cite{roche:2018}.\footnote{We use the convention, $\widetilde{O}(t)=O\left(t\cdot \log^{O(1)} t\right)$ , and when no basis is mentioned then the log is taken in the natural base $e$.}      
\begin{openproblem}[{\cite[Open Problem 3]{roche:2018}}]
	\label{op:exact}
	Given two sparse polynomials $f, g \in R[x]$, develop an algorithm which determines whether $g$ divides $f$ exactly, using
	$\widetilde{O}(T  \cdot \log{\deg{f}})$ ring and bit operations,\footnote{In the study of algorithms over rings, we distinguish between operations on ring elements, which may depend on the ring's representation, and operations on bit strings. For instance, ring operations include addition, multiplication, and division of coefficients, while bit operations correspond to addition and subtraction of exponent vectors, typically represented as bit strings. In practical scenarios involving domains like $\Z$ or $\C$, the computational cost of each ring operation on a computer is quasi-linear in the number of bits used to represent the relevant elements.} 
	where $T$ is the total number of terms of $f$, $g$, $(f \text{ quo } g)$ and $(f \text{ rem } g)$ (the quotient and remainder).
\end{openproblem}

Note that \autoref{op:exact} asks for an algorithm for the decision problem, whose running time may be exponential in the input size (as in the search problem).

    \subsection{Our Results}
        \label{sec:intro:results}

        Our  main result gives a new bound on the
            $\ell_2$-norm of $f/g$, for $f,g\in\C[x]$ such that $g$ divides $f$. This result improves the classical bounds of Gel'fond \cite{mahler:1962} and  Mignotte \cite{mignotte:1974} when the factor $g$ is sparse. In what follows we assume without loss of generality that $f(0)\neq 0$ as we can always remove powers of $x$ that divide the polynomials without affecting their coefficients. To slightly simplify the expression, we assume that $f(0)\neq 0$. If this is not the case, then $\deg{g}$ should be replaced by $\deg{g}-\ord{0}{g}$ throughout.
\begin{theorem}
            \label{thm:intro:norm-bound}
 	Let $f, g, h \in \C[x]$ such that $f = gh$ and $f(0)\neq 0$.
 	Denote $L(n) = n\log{n}$. Let
 			\[\Max
 	= \max \left\{ 
 	\begin{array}{c} 
 		\min\{\abs{\lc(g)},\abs{\tc(g)}\}\cdot \left(\frac{\|g\|_0}{2e}\right)^{\frac{\|g\|_0-1}{2}}\cdot \left(\frac {\deg{g}}{e}\right)^{\frac{\norm{g}_0}{2}}\\
 		\max\{\abs{\lc({g})},\abs{\tc({g})}\}\cdot \deg{g}
 	\end{array} 
 	\right\} .
 	\]
 Then,  
 \begin{equation*}
 	\norm{h}_2
 	\leq \frac{\sqrt{2}\norm{f}_1}{\Max{}}
 	\left({2\norm{g}_0 \cdot L\left( 12\norm{g}_0\deg{g} + 2\deg{f} \right)^2}\right)^{\norm{g}_0-1}
 	.
 \end{equation*}
        \end{theorem}
        
        \begin{remark}
        	We note that we must have a term that inversely depend on $|\lc(g)|$ or $|\tc(g)|$
since if $hg=f$ then for every constant $c$ we also have $(c\cdot h)\cdot (g/c)=f$. This scales $\|h\|_2$ by a factor of $c$. 
        \end{remark}
  
The following is an immediate corollary.
  
  \begin{corollary}\label{cor:intro:main}
  	Let
  	\begin{itemize}
  		\item $f, g, h \in \Z[x]$ such that $f = gh$ and $f(0)\neq 0$, or 
  		\item $f, g, h \in \C[x]$ such that $f = gh$, $f(0)\neq 0$,  and $|\lc(g)|,|\tc(g)|\geq 1$.
  	\end{itemize}  
	Then, it holds that
  	\begin{equation*}\label{eq:bound-Z}
 \norm{h}_2
 \leq \sqrt{\frac{\norm{g}_0}{\deg g}}\cdot \|f\|_1\cdot 
 \left(\widetilde{O}\left(\frac{\norm{g}_0^{2.5} \cdot  \mathrm{deg}^2{g} + \sqrt{\|g\|_0}\cdot  \mathrm{deg}^2{f}} {\sqrt{\deg{g}}}\right)^2\right)^{\norm{g}_0-1}
 . 	\end{equation*}
\  	In terms of bit representation size we have
  	\begin{equation}\label{eq:bound-bit-Z}
  		\log_2{\norm{h}_{\infty}} \leq \log_2{\norm{h}_{2}}\leq  \log_2{\norm{f}_{\infty}} + O\left(\|g\|_0\cdot \log_2{\deg{f}}\right).
  	\end{equation}
  \end{corollary}
  
%
    \begin{corollary}
        Let $f, g \in \C[x]$ such that $f = gh$ and $g$ is monic.
        Then it holds that
        \begin{equation*}
        \norm{h}_2
            \leq \frac{\|f\|_1}{\deg{g}}\cdot 
            \widetilde{O}\left(
                \norm{g}_0^{3} \mathrm{deg}^2{f}\right)^{\norm{g}_0-1}
            . 	\end{equation*}
    \end{corollary}
        Unlike the bounds in Equations~\eqref{eq:gelfonds-lemma}, \eqref{eq:mignotte-bound} and \eqref{eq:euclidean-division-bound},
            our bound on the height of $h$ is not exponential in the degree or number of coefficients of $h$, but rather in the sparsity of the factor $g$, which for sparse polynomials may be exponentially smaller. 
	%
%
        Thus,  when $\|g\|_0$ is small,
            this gives an exponential improvement over the bounds of Gel'fond and Mignotte (Theorems~\ref{thm:gelfond},~\ref{thm:Mignotte}). 
             In addition,       \autoref{thm:intro:norm-bound} gives an exponential separation between
            exact and non-exact division  since $\log_2{\norm{h}_{\infty}}$
            can be as large as $\Omega(\deg{f} \cdot \log_2{\norm{g}_{\infty}})$
            when allowing non-zero remainder.

		\autoref{thm:intro:norm-bound} 
		 also provides the following guarantee regarding the performance of the polynomial division algorithm in the exact case.
		
		\begin{theorem}\label{thm:div-alg}
			When the division is exact, the time complexity of the polynomial division algorithm 
			for polynomials $f, g \in \Z[x]$ is 
            \begin{equation*}
			    \widetilde{O}\left(
                    \norm{f/g}_0
                    \cdot
                    \norm{g}_0
                    \cdot \left(
                        \log{\norm{g}_{\infty}} + \log{\norm{f}_{\infty}} + \norm{g}_0\log{\deg{f}}
                    \right)
                \right).
			\end{equation*}
		\end{theorem}
		Note that when  $\norm{g}_0 = \poly(\log{\norm{f/g}_0})$
		or $\norm{g}_0 = \poly(\log(\norm{g}_\infty\cdot \|f\|_\infty\cdot \deg{f}))$, the bound in \autoref{thm:div-alg} is $\widetilde{O}(\norm{f/g}_0  (\log\deg{f} + \log H))$
		where $H$ is an upper bound on the heights of $f$ and $g$. This gives a positive answer to \autoref{op:exact} (in fact, for the more difficult search version), for this setting of parameters.

        Finally, combining the bound in \autoref{thm:intro:norm-bound} with Theorem 1.3 of \cite{giorgi2022sparse}
            we obtain the following result.
            \begin{corollary}\label{cor:rand-div}
            There exists a randomized algorithm that,
            given two sparse polynomials $f, g \in \Z[x]$,
            such that $g$ divides $f$,
            the algorithm outputs the quotient polynomial $f/g$, with probability at least $2/3$. The running time of the algorithm is
            \begin{equation*}
            	\widetilde{O}\left(
            	(\norm{f}_0 + \norm{g}_0 + \norm{f/g}_0)                   
            	\cdot \left(
            	\log{\norm{g}_{\infty}} + \log{\norm{f}_{\infty}} + \norm{g}_0\log{\deg{f}}
            	\right)
            	\right).
            \end{equation*}	
            \end{corollary}

    \subsection{Proof Outline}
        
        We first note that  \autoref{thm:div-alg} follows immediately from \autoref{thm:intro:norm-bound} and properties of the 
            polynomial division algorithm.
        The proof of \autoref{thm:intro:norm-bound} starts, similarly to \cite[Theorem 2.9]{giesbrecht-roche:2011}, by considering the discrete Fourier transform
            of the coefficient vector of  $h$. Parseval's identity gives
            \begin{equation*}
                \norm{h}_2^2 = \frac{1}{p}\sum_{0 \leq i < p}\abs{h(\omega^i)}^2
                ,
            \end{equation*}
            where $\omega$ is a  primitive $p$th root of unity for a prime $p>\deg(h)$.
        Thus, we conclude that $\norm{h}_2 \leq \abs{h(\theta)}$
            for some $p$th root of unity $\theta$.
        Since $\abs{g(\theta)}\cdot\abs{h(\theta)} = \abs{f(\theta)} \leq \norm{f}_1$, we conclude that
        \begin{equation*}
            \norm{h}_2 \leq \abs{h(\theta)} \leq \frac{ \norm{f}_1}{\abs{g(\theta)}}.
        \end{equation*}
        Hence, it is enough to lower bound the values of $g$ at roots of unity in order to upper bound $\|h\|_2$.
        To prove such a bound, we consider the restriction of $g$ to the unit circle, 
            $\tilde{g}(x) = g(e^{ix}): [0, 2\pi) \to \C$.
	    We prove by induction on the sparsity of $g$ that outside the neighborhood
            of a small set of bad points $B(g) \subset [0,2\pi)$, $g(e^{ix})$
            attains large values.
	    Indeed, when $\deg{g} = 1$ this is relatively simple to show
            as $g$ must be of the form $x - \alpha$ (after rescaling)
            and since primitive roots of unity of relatively prime orders
            are somewhat far from each other, we can find such prime $p$ so that 
            the values of $g$  on primitive $p$th roots of unity are not too small.
	    For higher degrees, we note that if the derivative of $g$ is not small
            in a large enough region then $g$ is large on most of that region.
        Thus, by moving further away from the bad set for the derivative of $g$, $B(g')$ 
	        we get that on this set $g$ always attains large values. 
	    We then prove by simple pigeonhole principle there exists a prime $p$
            such that all primitive $p$th roots of unity
            are far away from the bad set and conclude the result.

\section{Preliminaries}
    \label{sec:prelim}

    In this section, we set our notation,
        and state some basic facts that we shall later use.

    \subsection{General Notation}
        \label{sec:notation}
        Let $R$ be a unique factorization domain (UFD) and let $g,f \in R[x]$ be two polynomials.
	    We say that $g$ divides $f$, and denote it with $g\mid f$,
            when there is a polynomial $h\in R[x]$ 
            such that $g\cdot h=f$. We denote $g \nmid f$ when this is not the case.
	
	    Let $f=\sum_i f_i x^i$,. We shall use the following notation:
	        $\norm{f}_0$ denotes the sparsity of $f$ (the number of non-zero $f_i$'s);
            $\norm{f}_p = \left( \sum_{i}{\abs{f_i}^p} \right)^{1/p}$;
            if $R\subseteq \C$ then
            $\norm{f}_{\infty} = \max_i{\abs{f_i}}$ is the height of $f$;
            $\ord{\alpha}{f}$ denotes the multiplicity of a root $\alpha$
            in $f$, in particular,  $\ord{0}{f} = \max{\set{i: x^i \mid f}}$;
            $\lc(f) \in R$ denotes the leading coefficient in $f$,
            i.e., the coefficient of the highest degree term in $f$. Similarly, $\tc(f) \in R$, the trailing coefficient of $f$, denotes the  coefficient of the smallest degree term in $f$.

		For a polynomial $g$ we denote $g_0:=(g/x^{\ord{0}{g}})$ and $g_{\rev}:=x^{\deg{g}}\cdot g(1/x)$. 
		For example, if $g=\sum_{i=1}^{s}c_i\cdot x^{n_i}$ then $g_0= \sum_{i=1}^{s}c_i\cdot x^{n_i-n_1}$ and $g_{\rev}=\sum_{i=1}^{s}c_i\cdot x^{n_s-n_i}$.  Observe that $(g_{\rev})_0=(g_0)_{\rev}$ and that $g_0=(g_{\rev})_{\rev}$.
		
We denote the derivative of $g$ by $g'$.

        For any real or complex-valued function $f$, 
	        we define $V(f)$ to be the set of zeros of $f$,
            i.e., $V(f) = \set{\alpha \mid f(\alpha) = 0}$, over the corresponding field.
            For a complex-valued function $f$ we denote by $\Re{f}$ and $\Im{f}$ 
	        the real and imaginary parts of $f$, respectively.

\ignore{
        A polynomial with $t$ non-zero terms is called a $t$-nomial;
            we call a polynomial with two terms a binomial,
            a polynomial with three terms a trinomial,
            and a polynomial with five terms a pentanomial.
}

        For an integer $k \in \N$, we denote $[k] = \set{1, \dots, k}$.
        Throughout this paper, $\log$ refers to the natural logarithm.
        We use $\pi$ to denote the mathematical constant,
            while $\pi(\cdot)$ refers to the prime counting function.

    \subsection{Useful Facts}
        \label{sec:facts}

        \subsubsection{Complex Analysis}

            For general facts about complex analysis, see
            \cite{ahlfors-complex-analysis:2021,titchmarsh-theory-of-functions:2002}.
            For the definition of harmonic functions,
                see \cite[Chapter 4.6]{ahlfors-complex-analysis:2021}
                and \cite{axler:2013}.

            \begin{claim}
                \label{clm:analytic-function}
                Let $f: \C \to \R$ be a harmonic function.
                If $f(\omega) = 0$ for every $\omega$ such that $\abs{\omega} = 1$,
                    then $f = 0$.
            \end{claim}
            \begin{proof}
                Let,
                \begin{equation*}
                    D = \set{\omega \in \C : \abs{\omega} \leq 1}
                \end{equation*}
                be the unit disc.
		        From the Maximum Modulus Principle for harmonic functions
                    (see  \cite[Chapter 4.6.2, Theorem 21]{ahlfors-complex-analysis:2021})
                    and compactness of $D$ we get that  $f(\omega) = 0$ for all $\omega \in D$.
                It follows that  $f$  must vanish everywhere
                    \cite[Theorem 1.27]{axler:2013}.
            \end{proof}

            \begin{claim}
                \label{clm:complex-polynomial-roots}
                Let $f \in \C[x]$ be a polynomial.
                If $\Re{f}$
                    (or equivalently, $\Im{f}$, $\Re{f} + \Im{f}$ or $\Re{f} - \Im{f}$)
                    is non-zero,
                    then it has at most $2\deg{f}$ roots on the unit circle.
            \end{claim}
            \begin{proof}
                Let $S \subseteq \C$ be the set of roots of $\Re{f}$
                    on the unit circle.
                We can write $\Re{f}(x + iy) = u(x, y)$ where
                    $u: \R^2 \to \R$ is a bi-variate polynomial
                    of total degree $\deg{f}$.
                Hence,
                \begin{equation*}
                    |S| = |V(u) \cap V(x^2 + y^2 - 1)|
                    .
                \end{equation*}
                Since $f$ is analytic, $\Re{f}$ is harmonic.
                \autoref{clm:analytic-function} implies that $\Re{f}$
                    does not vanish on the unit circle.
                Thus, from the irreducibility of $x^2 + y^2 - 1$, we deduce that
                    $V(u)$ and $V(x^2 + y^2 - 1)$ have no component in common.
                From B\'ezout's theorem (see e.g., \cite[Theorem 18.3]{Harris-AG-book})
                    we conclude that 
                \begin{equation*}
                    |S| \leq \deg{u} \cdot \deg{(x^2 + y^2 - 1)} = 2\deg{f}
                    .
                    \qedhere
                \end{equation*}
            \end{proof}

            \begin{claim}
                \label{clm:eix-analysis}
                If $\abs{x - y} \leq \pi$,
                then $\abs{e^{ix} - e^{iy}} \geq \frac{2}{\pi} \abs{x - y}$.
            \end{claim}
            \begin{proof}
                Since
                \begin{equation*}
                    \abs{e^{ix} - e^{iy}} = \abs{e^{i(x - y)} - 1}
                    ,
                \end{equation*}
                    assume without loss of generality that $y = 0$.
                If $\abs{x} \leq \pi/2$, then
                \begin{equation*}
                        \abs{e^{ix} - 1}
                        = \sqrt{(1 - \cos{x})^2 + \sin^2{x}}
                        \geq \abs{\sin{x}}
                        \geq \frac{2}{\pi} \abs{x}
                        ,
                \end{equation*}
                where the last inequality is due to the fact that
                    $\sin$ is positive and concave on $[0, \pi]$,
                    while negative and convex on $[-\pi, 0]$.
                Similarly, if $\pi/2 \leq \abs{x} \leq \pi$, then
                \begin{equation*}
                        \abs{1 - e^{ix}}
                        = \sqrt{(1 - \cos{x})^2 + \sin^2{x}}
                        \geq \abs{1 - \cos{x}}
                        \geq \frac{2}{\pi} \abs{x}
                        ,
                \end{equation*}
                where, as before, the last inequality is due to the fact that
                    $1 - \cos$ is positive and concave on
                    both $[\frac{\pi}{2}, \frac{3\pi}{2}]$
                    and $[-\frac{3\pi}{2}, -\frac{\pi}{2}]$.
            \end{proof}

        \subsubsection{Number Theory}

            \begin{fact}[{\cite[Corollary 1]{rosser-schoenfeld:1962}}]
                \label{fct:prime-counting-lb}
                Let $\pi(n) = \abs{\set{p \leq n: p \text{ is prime}}}$
                    be the prime-counting function.
                Then $\pi(n) \geq \frac{n}{\log{n}}$ for $n \geq 17$.
            \end{fact}

	            \begin{fact}[{\cite[Corollary 3]{rosser-schoenfeld:1962}}]
                \label{fct:prime-counting-2n}
                Let $n \in \N$ such that $n \geq 2$.
                Then, there exists at least one prime in $(n, 2n]$.
            \end{fact}

\section{A Bound on the \texorpdfstring{$\ell_2$}{l2}-norm of the Factor Polynomial}
    \label{sec:norm-bound}

    In this section, we prove \autoref{thm:intro:norm-bound}.
    We do so by analyzing evaluations of polynomials at roots of unity of high orders.

    The following lemma uses simple Fourier analysis to bound
        the norm of the quotient polynomial
        in a similar fashion to Theorem 2.9 of \cite{giesbrecht-roche:2011}.

    \begin{lemma}
        \label{lma:dft}
        Let $f, g, h \in \C[x]$ such that $f = gh$.
        Then,
        \begin{equation*}
            \norm{h}_2 \leq \sqrt{2}\norm{f}_1 \cdot
            \max_{\substack{1 \neq \theta \in \C \\ \theta^p = 1}}
            {\frac{1}{\abs{g(\theta)}}}
        \end{equation*}
        where $p$ is a prime such that $p > 2\deg{f}$ and $g$ does not vanish at any primitive root of unity of order $p$.
    \end{lemma}
    \begin{proof}
        Recall that the $p \times p$ Discrete Fourier Transform (DFT) matrix is
        \begin{equation*}
            \frac{1}{\sqrt{p}}\left( \omega^{jk} \right)_{0 \leq j, k < p}
        \end{equation*}
        where $\omega \in \C$ is a primitive $p$th root of unity.
        Since the DFT matrix is unitary, applying it to the coefficient vector of $h$ we conclude that
        \begin{equation*}
            \norm{h}_2^2
            = \sum_{0 \leq i \leq \deg{h}}{\abs{h_i}^2}
            = \frac{1}{p}\sum_{0 \leq i < p}{\abs{h(\omega^i)}^2}.
        \end{equation*}
        From
        \begin{equation*}
            \abs{h(\omega^0)}^2
            = \abs{h(1)}^2
            \leq \norm{h}_1^2
            \leq \norm{h}_2^2\deg{h}
            \leq \norm{h}_2^2\deg{f}
            ,
        \end{equation*}
		we obtain, by the choice of $p$,
        \begin{equation*}
            \frac{1}{p}\sum_{0 < i < p}{\abs{h(\omega^i)}^2}
            = \norm{h}_2^2 - \frac{\abs{h(\omega^0)}^2}{p} \geq
            \norm{h}_2^2 - \frac{\deg{f}}{p}\norm{h}_2^2
            > \frac{1}{2}\norm{h}_2^2
            .
        \end{equation*}
        In other words,
        the average value of $\abs{h(\omega^i)}^2$ for $0 < i < p$
        is larger than $\frac{1}{2}\norm{h}_2^2$.
        Hence, there exists a $p$th root of unity $\theta \neq 1$ such that
        \begin{equation*}
            \norm{h}_2< \sqrt{2}\abs{h(\theta)}
            .
        \end{equation*}
        Since $f = gh$ we conclude that
        $\abs{f(\theta)} = \abs{g(\theta)}\abs{h(\theta)}$.
        As $\abs{\theta} = 1$, $\abs{f(\theta)} \leq \norm{f}_1$.
        Hence,
        \begin{equation*}
            \norm{h}_2< \sqrt{2}\abs{h(\theta)}
            = \sqrt{2}\frac{\abs{f(\theta)}}{\abs{g(\theta)}}
            \leq \sqrt{2}\frac{\norm{f}_1}{\abs{g(\theta)}}
            .
        \end{equation*}
        By taking the maximum over all $p$-th roots of unity $\theta \neq 1$ we get
        \begin{equation*}
            \norm{h}_2 \leq \sqrt{2}\norm{f}_1 \cdot
            \max_{\substack{1 \neq \theta \in \C \\ \theta^p = 1}}
            {\frac{1}{\abs{g(\theta)}}}
        \end{equation*}
        as claimed.
    \end{proof}

We first consider the special case where $g$ is linear.
We show that there exists a point $\tilde{\alpha}$
on the unit circle such that $g$ attains
large values on any point on the unit circle
that is far enough from $\tilde{\alpha}$.
To conclude, we prove that there exists a prime $p$ such that every
$p$th root of unity is far from $\tilde{\alpha}$.


\begin{lemma}
	\label{lma:good-frequency}
	Let $\alpha_1, \dots, \alpha_k \in \R$.
	Then, any set of $k + 1$ prime numbers $P$ contains $p\in P$ such that
	$\abs{a/p-\alpha_i} > \frac{1}{2\left( \max{P} \right)^2}$
	for every integer $0 < \abs{a} < p$ and $i \in [k]$.
\end{lemma}
\begin{proof}
	Assume for the sake of contradiction that the statement is false.
	Thus, for every $p \in P$ there exists $i \in [k]$ such that
	$\abs{a/p - \alpha_i}\leq \frac{1}{2\left( \max{P} \right)^2}$
	for some $0 < a < p$.
	By the Pigeonhole Principle,
	there exists $i \in [k]$ such that for two different primes
	$p, q \in P$, $0 < a < p$ and $0 < b < q$ we have
	\begin{equation*}
		\begin{aligned}
			\abs{a/p - \alpha_i} \leq \frac{1}{2\left( \max{P} \right)^2}
			&\quad\text{and}&
			\abs{b/q - \alpha_i} \leq \frac{1}{2\left( \max{P} \right)^2}
			.
		\end{aligned}
	\end{equation*}
	This, however, leads to a contradiction
	\begin{equation*}
		\frac{1}{\left( \max{P} \right)^2}
		\geq \abs{\frac{a}{p} - \alpha_i} + \abs{\alpha_i - \frac{b}{q} }
		\geq\abs{\frac{a}{p} - \frac{b}{q}}
		= \abs{\frac{ap - bq}{pq}}
		\geq 1/pq> \frac{1}{\left( \max{P} \right)^2}
		.\qedhere
	\end{equation*}
\end{proof}

%
%

    \begin{proposition}
        \label{prp:linear-division-bound}
        Let $f \in \C[x]$ with a root $\alpha$.
        Then
        $\norm{f/(x - \alpha)}_2 \leq  64\sqrt{2} \norm{f}_1  \deg^2{f} <  100 \norm{f}_1  \deg^2{f}$.
    \end{proposition}
    \begin{proof}
        \autoref{lma:dft} applied to $g=(x-\alpha)$ and $h=f/g$ implies that
        \begin{equation}
            \label{eq:fft-for-linear-function}
            \norm{f/(x - \alpha)}_2 \leq \sqrt{2}\norm{f}_1 \cdot
            \max_{\substack{1 \neq \theta \in \C \\ \theta^p = 1}}{
                \frac{1}{\abs{\theta - \alpha}}
            }
            ,
        \end{equation}
        for any prime $p$ such that $p > 2\deg{f}$ and $\alpha$ is not a primitive root of unity of order $p$.
        We prove that there exists a suitable $p$ such that
            $\abs{\theta - \alpha}$ in \autoref{eq:fft-for-linear-function}
            must be at least $(64\deg^2{f})^{-1}$.
        First, observe that if $\alpha$ is $(64\deg^2{f})^{-1}$-far
            from any point on the unit circle,
            then we are done.
        So, assume that there exists a point $\tilde{\alpha}$
            on the unit circle such that
            $\abs{\alpha - \tilde{\alpha}} < (64\deg^2{f})^{-1}$,
            and let $c \in [-\pi, \pi)$ such that $e^{ic} = \tilde{\alpha}$.
        
        By \autoref{fct:prime-counting-2n},
            the interval $(2\deg{f}, 8\deg{f}]$ contains at least 2 primes.
        Thus, by \autoref{lma:good-frequency}, there exists
            a prime $p$ in this interval such that
        \begin{equation*}
            \abs{\frac{2\pi a}{p} - c}
            = 2\pi\abs{\frac{ a}{p} - \frac{c}{2\pi}}
            > \frac{\pi}{64\deg^2{f}}
            ,
        \end{equation*}
        for any integer $0 < \abs{a} < p$.
        By the periodicity of $e^{ix}$,
            any $p$th root of unity $\theta \neq 1$
            can be expressed as $\theta=e^{i\frac{2\pi a}{p}}$
            such that $a \in \Z$, $0 < \abs{a} < p$
            and $\abs{\frac{2\pi a}{p} - c} \leq \pi$.
        Hence, by \autoref{clm:eix-analysis} it holds that
        \begin{equation*}
            \abs{\theta - \tilde{\alpha}}
            \geq \frac{2}{\pi}\abs{\frac{2\pi a}{p} - c}
            > \frac{1}{32\deg^2{f}},
        \end{equation*}
        for any $p$th root of unity $\theta \neq 1$.
        The triangle inequality implies
        \begin{equation*}
            \abs{\theta - \alpha}
            \geq \abs{\theta - \tilde{\alpha}} - \abs{\tilde{\alpha} - \alpha}
            > \frac{1}{64\deg^2{f}}
            .
        \end{equation*}
        Finally, by using this estimate in \autoref{eq:fft-for-linear-function}
            we conclude that
        \begin{equation*}
            \norm{f/(x - \alpha)}_2 \leq \sqrt{2} \norm{f}_1 \cdot64 \deg^2{f}< 100 \norm{f}_1 \deg^2{f}
            ,
        \end{equation*}
        which completes the proof.
    \end{proof}



We prove \autoref{thm:intro:norm-bound} in its generality by induction on $\|g\|_0$. Similarly to the proof outline of \autoref{prp:linear-division-bound},
we will show that $g$ attains large values everywhere, with the possible exception of a small neighborhood of a small number of points. We start with two simple claims.

    \begin{claim}
        \label{clm:real-analysis}
        Let $f: \R \to \R$ be a continuously differentiable function, and
        $\delta > 0$. Let $I = [a, b]$ be  an interval in which both $f$ and $f'$ do not change sign. That is, each of them is either non-negative or non-positive in $I$.
        Then, $|f(\beta)| \geq \delta\min_{\alpha \in I}|f'(\alpha)|$
        for every $\beta \in [a + \delta, b - \delta]$.
    \end{claim}
    \begin{proof}
        If $b - a < 2\delta$ there is nothing to prove.
        Assume without loss of generality that $f \geq 0$ in $I$,
        otherwise analyze $-f$.
        If $f' \geq 0$ in $I$, then $f$ is non-decreasing in $I$ and so
        for every $\beta \in {[a + \delta, b - \delta]}$ it holds that
        \begin{equation*}
                f(\beta)
                \geq f(\beta) - f(a)= \int_{a}^\beta{f'(y)dy}\geq (\beta - a)\min_{\alpha \in I}f'(\alpha\geq \delta\min_{\alpha \in I}|f'(\alpha)|.
        \end{equation*}
        Similarly, if $f' \leq 0$ in $I$,
        then for every $\beta \in {[a + \delta, b - \delta]}$ it holds that
        \begin{equation*}
                f(\beta)  \geq f(\beta) - f(b)= \int_{\beta}^{b}{-f'(y)dy}\geq (b - \beta)\min_{\alpha \in I}|f'(\alpha)| \geq \delta\min_{\alpha \in I}|f'(\alpha)|,
        \end{equation*}
        as claimed.
    \end{proof}

    \begin{claim}
        \label{clm:plane-analysis}
        Let $f: \R \to \C$ be a continuously differentiable function,
            $\delta > 0$ and denote $f_1 = \Re{f}$, $f_2 = \Im{f}$.
        Let $I = [a, b]$ be an interval in which each function among
            $f_1, f'_1, f_2, f'_2, f_1' - f_2'$
            and $f_1' + f_2'$ is either non-negative or non-positive.
        Then, $\abs{f(\beta)} \geq \frac{\delta}{\sqrt{2}}\min_{\alpha \in I}|f'(\alpha)|$
            for every $\beta \in [a + \delta, b - \delta]$.
    \end{claim}
    \begin{proof}
        Our assumption implies that $\abs{f_1'} - \abs{f_2'}$
            is either non-negative or non-positive in $I$, as well.
        Without loss of generality assume $\abs{f_1'} \geq \abs{f_2'}$ in $I$.
        For every $\beta \in I$, $f'(\beta)=f'_1(\beta) + i\cdot f'_2(\beta)$.
        Hence,
        \begin{equation*}
            2\abs{f'_1(\beta)}^2\geq \abs{f'_1(\beta)}^2 + \abs{f'_2(\beta)}^2
            = \abs{f'(\beta)}^2
            \geq \min_{\alpha \in I}\abs{f'(\alpha)}^2
            .
        \end{equation*}
        Thus, $\abs{f_1'(\beta)} \geq \frac{1}{\sqrt{2}}\min_{\alpha \in I}\abs{f'(\alpha)}$
            for every $\beta \in I$.
        The result follows by \autoref{clm:real-analysis}
            and the fact that $\abs{f(\beta)} \geq \abs{f_1(\beta)}$.
    \end{proof}

	We shall need the following definition in order to state our result.
	
	\begin{definition}\label{def:d}
		For a polynomial $g(x)=\sum_{i=1}^{s}c_ix^{n_i}$ with $n_1<n_2<\ldots<n_s$ and $\prod_{i=1}^{s}c_i\neq 0$ we define $d(g):= \prod_{i=1}^{s-1}(n_s-n_i)$. When $s=1$ we set $d(g)=1$.
	\end{definition}

	In the next claims 	we shall use the notation introduce in \autoref{sec:notation}: $g_0=g/x^{\ord{0}{g}}$ and $g_{\rev}=x^{\deg{g}}\cdot g(1/x)$.

	\begin{claim}\label{cla:dg'}
		Let $g$ as in \autoref{def:d}. We have $d(g)=d(g_0)=\deg{g_0} \cdot d((g_0)')$.
	\end{claim}
	\begin{proof}
		The claim follow by easy calculations:
		\[
		d(g_0) =  \prod_{i=1}^{s-1}\left((n_s-n_1)-(n_i-n_1)\right)= \prod_{i=1}^{s-1}(n_s-n_i)=d(g).
		\]
		Next, observe that \[(g_0)'=(\sum_{i=1}^{s}c_ix^{n_i-n_1})'=\sum_{i=2}^{s}(n_i-n_1)c_ix^{n_i-n_1-1}.\] 
		Hence,
		\begin{align*}
			\deg{g_0}\cdot d((g_0)') &= (n_s-n_1)\cdot \prod_{i=2}^{s-1}\left((n_s-n_1-1)-(n_i-n_1-1)\right)\\ &=\prod_{i=1}^{s-1}(n_s-n_i)=d(g).\qedhere
			\end{align*}
	\end{proof}
	
	\begin{claim}\label{cla:dg}
		Let $g$ as in \autoref{def:d}.
		Then, $d(g_{\rev})=\prod_{i=2}^{s}(n_i-n_1)$.
	\end{claim}
	\begin{proof}
		The case $s=1$ is clear. 
		Let $m_{s-i+1}=n_s-n_i$. Clearly, $0=m_1<m_2<\ldots < m_s=n_s-n_1$. Then, $g_{\rev}=\sum_{i=1}^{s}c_ix^{n_s-n_i}=\sum_{i=1}^{s}c_{s-i+1}x^{m_i}$ and we get that 
		\[d(g_{\rev})=\prod_{i=1}^{s-1}(m_s-m_i)=\prod_{i=1}^{s-1}\left((n_s-n_1)-(n_s-n_{s-i+1})\right)=\prod_{i=2}^{s}(n_i-n_1).\qedhere
		\]
	\end{proof}
	
	\begin{corollary}\label{cor:max}
			Let  $g$ as in  \autoref{def:d}. Then, 
		\[
		\max \{d(g),d(g_{\rev})\}\geq (n_s - n_1 - 1)^{s/2}.
\]
	\end{corollary}
	 \begin{proof}
	 	Note that 
	 	\begin{equation}\label{eq:dd}
	 		d(g) d(g_{\rev})= \prod_{i=1}^{s-1}(n_s-n_i)\cdot \prod_{i=2}^{s}(n_i-n_1)=(n_s-n_1)^2\cdot \prod_{i=2}^{s-1}\left((n_s-n_i) (n_i-n_1)\right).
	 		\end{equation}
		Since $(n_s - z)(z - n_1) \geq n_s - n_1 - 1$ for any integer $n_1 < z < n_s$,
		we can conclude that,
		\begin{equation*}
			d(g)d(g_{rev}) \geq (n_s - n_1)^2(n_s - n_1 - 1)^{s - 2} \geq (n_s - n_1 - 1)^{s}.\qedhere
		\end{equation*}
	\end{proof}
	
	We can improve this estimate by a slightly more careful analysis.
	
	\begin{claim}\label{cla:maxd}
		Let  $g$ as in  \autoref{def:d}. If $n_s-n_1>1$  then
		\[
		\max \{d(g),d(g_{\rev})\}\geq  \left(\frac{s}{2e}\right)^{\frac {s}{2} -1}\cdot \left(\frac{n_s-n_1}{e}\right)^{ s/2}.
		\]
	\end{claim}
	\begin{proof}
		As in the proof of \autoref{cor:max} we consider \eqref{eq:dd}.
		Since the minimum of the expression $(n_s - z)(z - n_1)$ for an integer $n_1 < n_1+a \leq z \leq n_s-b < n_s$
		is attained at $z =n_1+a$ if $a\leq b$ and at $n_s-b$ otherwise,
		and given that $n_1,\ldots,n_s$ are distinct,
		we  conclude by  induction that, 
\begin{equation*}
\begin{aligned}
&\max\{d(g),d(g_{\rev}) \}\geq \sqrt{ d(g)\cdot d(g_{\rev})}\\
&\geq (n_s-n_1)\cdot \sqrt{ \prod_{i=2}^{\lfloor \frac{s}{2}\rfloor} (n_s-n_1-i+1)(i-1)\cdot \prod_{i=1}^{\lceil \frac{s}{2}\rceil-1}(n_s-n_1-i)i}\\
&\geq \left(\left\lfloor \frac{s}{2}\right\rfloor -1\right)!\cdot\left( \frac{(n_s - n_1)!}{(n_s - n_1 - \floor{s / 2})!}\right)\cdot \left(\left(n_s-n_1-\frac{s-1}{2}\right)\left(\frac{s-1}{2}\right)\right)^{\frac{\ceil{s/2}-\floor{s/2}}{2}}\\
&\geq^{(*)} \left(\frac{s-3}{2e}\right)^{\frac {s}{2} -1}\cdot \left( \left(\frac{n_s-n_1}{e}\right)^{n_s-n_1}\cdot \left(\frac{e}{n_s-n_1-\lfloor s/2\rfloor} \right)^{n_s-n_1-\lfloor s/2\rfloor}  \right)\cdot \left(n_s-n_1-\frac{s-1}{2}\right)^{\frac{\ceil{s/2}-\floor{s/2}}{2}}\\
&\geq  \left(\frac{s-3}{2e}\right)^{\frac {s}{2} -1}\cdot  \left(\frac{n_s-n_1}{e}\right)^{\lfloor s/2\rfloor}  \cdot \left(1+\frac{\lfloor s/2\rfloor}{n_s-n_1-\lfloor s/2\rfloor} \right)^{n_s-n_1-\lfloor s/2\rfloor} \cdot \left(n_s-n_1-\frac{s-1}{2}\right)^{\frac{\ceil{s/2}-\floor{s/2}}{2}}
\\
&\geq\left(\frac{s-3}{e}\right)^{\frac {s}{2} -1}\cdot \left(\frac{n_s-n_1}{e}\right)^{\lfloor s/2\rfloor}\cdot \left(n_s-n_1-\frac{s-1}{2}\right)^{\frac{\ceil{s/2}-\floor{s/2}}{2}}\\
&\geq \left(\frac{s-3}{e}\right)^{\frac {s}{2} -1}\cdot \left(\frac{n_s-n_1}{e}\right)^{ s/2}, 
\end{aligned}
\end{equation*}
where to deduce inequality $(*)$ we applied Robbins' lower bound for the factorial function \cite{Robbins55} and some simple calculations. 
We note that for $s< 6$ and $n_s-n_1>1$ \autoref{cor:max} implies that 
	\[
\max \{d(g),d(g_{\rev})\}\geq (n_s - n_1 - 1)^{s/2}\geq \left(\frac{n_s-n_1}{e}\right)^{ s/2}.
\]
Consequently, if  $n_s-n_1>1$  then
	\[
\max \{d(g),d(g_{\rev})\}\geq  \left(\frac{s}{2e}\right)^{\frac {s}{2} -1}\cdot \left(\frac{n_s-n_1}{e}\right)^{ s/2}. \qedhere
\]
	\end{proof}

    In what comes next, we say that two points $\alpha$ and $\beta$
    are $\delta$-far from each other if $\abs{\alpha - \beta} \geq \delta$.
    Similarly, we say that a point is $\delta$-far from a set $S$
    if it is $\delta$-far from any point in $S$.

    \begin{lemma}
        \label{lma:inductive-lemma}
        Let $g \in \C[x]$ be a monic polynomial.  
        Then, there exists a set $B(g) \subseteq [0, 2\pi)$,
        satisfying $|B(g)|\leq  12(\norm{g}_0 - 1)\cdot(\deg{g} - \ord{0}{g})$,
        such that for every   $\delta>0$, and every $\alpha \in [0, 2\pi)$ which is $(\norm{g}_0 - 1)\delta$-far from $B(g)$,         it holds that
        $\abs{g(e^{i\alpha})} \geq d(g)\left(\frac{\delta}{\sqrt{2}} \right)^{\norm{g}_0-1}$.       
    \end{lemma}
    \begin{proof}
        
        We prove the lemma by induction on $s=\norm{g}_0$. 
        The claim holds for $s=1$ since  clearly $|g(e^{i\alpha})| = 1$ for any $\alpha$, when $g(x)=x^m$. 
        
        For the inductive step, suppose that the claim holds for $s$,
        and let $g \in \C[x]$ be of sparsity $s + 1$. Since $\abs{(g/x^{\ord{0}{g}})(\omega)} = \abs{g(\omega)}$
        for every $\omega$ on the unit circle we can consider $g_0=(g/x^{\ord{0}{g}})$ instead of $g$, which is of degree $\deg{g_0}=\deg{g}-\ord{0}{g}$. 
        
        Let $f: [0, 2\pi) \to \C$ be defined as $f(x) = g_0(e^{ix})$.
        Denote with 
        $$\hat{g}(x)=\frac{g_0'(x)}{\deg{g_0}}\, ,$$
        the monic polynomial which is a scalar multiple of  $g_0'$.
        By \autoref{cla:dg'}, $d(\hat{g})=d(g_0')$. Let 
    	\begin{equation} \label{eq:H}
            \begin{aligned}
                H &= \set{
                    \Re f, \Re f',
                    \Im f, \Im f',
                    \Re f' - \Im f',
                    \Re f' + \Im f'
                } \text{ and }\\
                B(g) &= B(\hat{g}) \cup \set{
                    \beta \in [0, 2\pi):
                    h(\beta) = 0
                    \text{ for a non-zero } h \in H
                }.
            \end{aligned}
    	\end{equation}
        By \autoref{clm:complex-polynomial-roots},
        each non-zero function in $H$ has at most $2\deg{g_0}$
        zeros in $[0, 2\pi)$.
        Since $\hat{g}(x)$ is monic and has sparsity $s$,
         we conclude by the inductive assumption that
        \begin{equation*}
            \abs{B(g)} \leq 12(s - 1)\deg{\hat{g}} + 12\deg{g_0} < 12s\deg{g_0}
            .
        \end{equation*}
        
        Let $I = [a, b] \subseteq [0, 2\pi)$
        be an interval that is $(s - 1)\delta$-far from $B(g)$.
        By the definition of $B(g)$, each functions in the set $H$, 
            as defined in Equation~\eqref{eq:H}, does not change its sign within $I$, since 
            any sign change  would have to go through a zero of the function.
        As $f'(x) = ie^{ix}\cdot {g_0'}(e^{ix})$, \autoref{clm:plane-analysis} implies that
            for every $\beta \in [a + \delta, b - \delta]$ it holds that
        \begin{equation}\label{eq:induct}
        	\begin{aligned}
        		\abs{g(e^{i\beta})}&=\abs{g_0(e^{i\beta})}= \abs{f(\beta)}
        		\geq \frac{\delta}{\sqrt{2}}\min_{\alpha \in I}\abs{f'(\alpha)}
        		\\
        		&= \frac{\delta}{\sqrt{2}}\min_{\alpha \in I}\abs{g_0'(e^{i\alpha})}= \frac{\delta\cdot \deg{g_0}}{\sqrt{2}}\min_{\alpha \in I}\abs{\hat{g}(e^{i\alpha})}
        		.
        	\end{aligned}	
        \end{equation}
        By definition, $B(\hat{g})\subseteq B(g)$ and as each $\beta \in I$
            is at least $(s-1)\delta$-far from $B(g)$,
            we get from the induction hypothesis applied to $\hat{g}$, Equation~\eqref{eq:induct} and 
            \autoref{cla:dg'} 
            \begin{equation*}
            	\begin{aligned}
            	\abs{g(e^{i\beta})} &\geq \frac{\delta\cdot \deg(g_0)}{\sqrt{2}}\min_{\alpha \in I}\abs{\hat{g}(e^{i\alpha})}\\&\geq  \frac{\delta\cdot \deg(g_0)}{\sqrt{2}}\cdot d(\hat{g}) \left( \frac{\delta}{\sqrt{2}} \right)^{s-1}\\&= \frac{\delta\cdot \deg(g_0)}{\sqrt{2}}\cdot d({g_0'}) \left( \frac{\delta}{\sqrt{2}} \right)^{s-1}\\
            	&= d(g)\left( \frac{\delta}{\sqrt{2}} \right)^s.
            	\end{aligned}
            \end{equation*}
        In other words, every $\beta$ that is $s\delta$
            far from the set $B(g)$ satisfies
            $\abs{g(e^{i\beta})} \geq d(g)\left( \frac{\delta}{\sqrt{2}} \right)^s$,
            which proves the step of the induction.
     \end{proof}

	\begin{corollary}\label{cor:main}
		        Let $g \in \C[x]$ be a polynomial.  
		Then, there exists a set $B(g) \subseteq [0, 2\pi)$,
		of size $|B(g)|\leq 12(\norm{g}_0 - 1)(\deg{g}-\ord{0}{g})$,
		such that for every   $\delta>0$ and every $\alpha \in [0, 2\pi)$
		which is $(\norm{g}_0 - 1)\delta$-far from $B(g)$,
		it holds that
		\[\abs{g(e^{i\alpha})} \geq \max\{\abs{\lc({g})}\cdot d({g}),\abs{\tc({g})}\cdot d(g_{\rev})\}\cdot \left(\frac{\delta}{\sqrt{2}} \right)^{\norm{g}_0-1}.\]      
	\end{corollary}
	
	\begin{proof}
		Let $g_{\rev}=x^{\deg{g}}g(1/x)$.  
		Since for every $\alpha$, 	$\abs{g(e^{i\alpha})}=	\abs{g_{\rev}(e^{-i\alpha})}$, we can consider either $g$ or $g_{\rev}$ in order to prove the claim (replacing $B(g)$ with $B(g_{\rev})$ if required). 
		Note that $g/\lc(g)$ and $g_{\rev}/\lc(g_{\rev})=g_{\rev}/\tc(g)$ are monic polynomials. From \autoref{lma:inductive-lemma} we have (where $\alpha$ is $\delta$-far from either $B(g)$ or $B(g_{\rev})$, respectively)
		\begin{equation*}
			\begin{aligned}
					\abs{\frac{g(e^{i\alpha})}{\lc(g)}} \geq d(g)\left(\frac{\delta}{\sqrt{2}} \right)^{\norm{g}_0-1} &\Rightarrow   \abs{g(e^{i\alpha})} \geq \abs{\lc(g)}\cdot d(g)\left(\frac{\delta}{\sqrt{2}} \right)^{\norm{g}_0-1}\\
					\abs{\frac{g_{\rev}(e^{i\alpha})}{\lc(g_{\rev})}}\geq d(g_{\rev})\left(\frac{\delta}{\sqrt{2}} \right)^{\norm{g}_0-1} &\Rightarrow   \abs{{g}(e^{i\alpha})} \geq \abs{\tc({g})}\cdot d(g_{\rev})\left(\frac{\delta}{\sqrt{2}} \right)^{\norm{g}_0-1}.\qedhere
			\end{aligned}
		\end{equation*}
	\end{proof}

%


    \begin{lemma}
        \label{lma:evaluation-on-root-of-unity-lb}
        Let $g \in \C[x]$.
        Let $p_{\min{}} \in \N$ and $L(n) = n\log{n}$.
        Set $p_{\max{}} = \lceil 2L(p_{\min{}} + 12\norm{g}_0(\deg{g}-\ord{0}{g}))\rceil$.
        Then, there exists a prime $p \in (p_{\min{}}, p_{\max{}}]$
        such that
        \begin{equation*}
        		\abs{g(\omega)} \geq  \max\{\abs{\lc({g})}\cdot d({g}),\abs{\tc({g})}\cdot d(g_{\rev})\}\cdot  \left(\frac{\pi}{\sqrt{2}\cdot\norm{g}_0\cdot p_{\max{}}^2}\right)^{\norm{g}_0-1}, 
        \end{equation*}
        for any $p$th root of unity $\omega \neq 1$.
    \end{lemma}
    \begin{proof}
        If $\norm{g}_0 = 1$, then the statement is trivial
        since $\abs{g(\omega)} = \abs{\lc(g)}=\abs{\tc(g)}$ for every root of unity $\omega$.
        Hence, we can assume that $\norm{g}_0 \geq 2$ and $\deg{g} \geq 1$.
        
        Let $t = p_{\min{}} + 12\norm{g}_0(\deg{g} - \ord{0}{g})$, and let $\pi(n)$ be the prime counting function.
        Our assumption implies that $t \geq 24$.
        By \autoref{fct:prime-counting-lb} we conclude that
        \begin{equation*}
            \begin{aligned}
                \pi(p_{\max{}})
                &= \pi(\lceil 2t\log{t}\rceil ) \\
                &\geq \frac{2t\log{t}}{
                    \log{2} + \log{t} + \log{\log{t}}
                }\\
                &> \frac{2t\log{t}}{2\log{t}}
                = t
                .
            \end{aligned}
        \end{equation*}
        Let $B(g)$ be as in \autoref{lma:inductive-lemma} and 
        \begin{equation*}
            P = \set{ p_{\min{}} < p \leq p_{\max{}} : p \text{ is prime} }
            .
        \end{equation*}
        We bound the size of $P$
        by counting the primes in the interval $(p_{\min{}}, p_{\max{}}]$:
        \begin{equation*}
            \abs{P}
            \geq \pi(p_{\max{}}) - p_{\min{}}
            > t - p_{\min{}} = 12\norm{g}_0(\deg{g} - \ord{0}{g}) \geq \abs{B(g)}
            .
        \end{equation*}
        By \autoref{lma:good-frequency},
        there exists $p \in P$ such that
        for any integer $0 < a < p$ and any $\alpha \in B(g)$
        \begin{equation*}
            \abs{a/p - \alpha/2\pi} \geq \frac{1}{2p_{\max{}}^2}
            .
        \end{equation*}
        Thus,
        \begin{equation*}
            \abs{2\pi\frac{a}{p} - \alpha} \geq \frac{\pi}{p_{\max{}}^2}
            .
        \end{equation*}
        In particular, every non-trivial $p$th root of unity is  $\frac{ \pi}{ p_{\max{}}^2}$-far from $B(g)$.
        Finally, by applying \autoref{cor:main}
        with $\delta = \frac{ \pi}{\|g\|_0\cdot p_{\max{}}^2}$,
        we conclude that for any $p$th root of unity $\omega \neq 1$,
        \begin{equation*}
            \abs{g(\omega)}
                \geq   \max\{\abs{\lc({g})}\cdot d({g}),\abs{\tc({g})}\cdot d(g_{\rev})\}\cdot  \left(\frac{\pi}{\sqrt{2}\cdot \norm{g}_0\cdot p_{\max{}}^2}\right)^{\norm{g}_0-1}
            .
            \qedhere
        \end{equation*}
    \end{proof}

    \begin{theorem}
        \label{thm:quotient-coefficient-bound}
        Let $f, g, h \in \C[x]$ such that $f = gh$.
        Denote $L(n) = n\log{n}$.
        Then,  for $\Max :=  \max\{\abs{\lc({g})}\cdot d({g}),\abs{\tc({g})}\cdot d(g_{\rev})\}$
	\begin{equation*}
		\norm{h}_2
		\leq \frac{\sqrt{2}\norm{f}_1}{\Max{}}
		\left({2\norm{g}_0 \cdot L\left( 12\norm{g}_0(\deg{g} - \ord{0}{g}) + 2\deg{f} \right)^2}\right)^{\norm{g}_0-1}
		.
	\end{equation*}
    \end{theorem}
    \begin{proof}
        Let $p_{\min{}} = 2\deg{f}$ and $p_{\max{}}$ as in \autoref{lma:evaluation-on-root-of-unity-lb}. 
        From \autoref{lma:evaluation-on-root-of-unity-lb}
        we conclude that there exists a prime $p \in (p_{\min{}}, p_{\max{}}]$ such that
        \begin{equation*}
            \abs{g(\omega)}
            \geq \Max \cdot \left( \frac{\pi }{\sqrt{2}\cdot \norm{g}_0\cdot (\ceil{2L(12\norm{g}_0(\deg{g} - \ord{0}{g}) + 2\deg{f})})^2} \right)^{\norm{g}_0-1}
        \end{equation*}
        for any primitive $p$th root of unity $\omega \neq 1$.
        The theorem then follows from \autoref{lma:dft}.
    \end{proof}

To conclude the proof we prove a lower bound on $\Max =  \max\{\abs{\lc({g})}\cdot d({g}),\abs{\tc({g})}\cdot d(g_{\rev})\}$.

\begin{lemma}
	\label{lem:max}
	Let $g \in \C[x]$.
	Let $\Max=\max\{\abs{\lc({g})}\cdot d({g}),\abs{\tc({g})}\cdot d(g_{\rev})\}$. Then,
	\[\Max
	\geq \max \left\{ 
	\begin{array}{c} 
		\min\{\abs{\lc(g)},\abs{\tc(g)}\}\cdot \left(\frac{\|g\|_0}{2e}\right)^{\frac{\|g\|_0}{2}-1}\cdot \left(\frac {\deg{g}-\ord{0}{g}}{e}\right)^{\frac{\norm{g}_0}{2}}\\
		\max\{\abs{\lc({g})},\abs{\tc({g})}\}\cdot (\deg{g}-\ord{0}{g})
	\end{array} 
	\right\} .
	\]
\end{lemma}

\begin{proof}
	Note that $d(g),d(g_{\rev})\geq (\deg{g}-\ord{0}{g})$ and therefore 
	\[\Max \geq \max\{\abs{\lc({g})},\abs{\tc({g})}\}\cdot(\deg{g}-\ord{0}{g}).
	\]
	In addition, since $\max\{\abs{\lc(g)}\cdot d(g),\abs{\tc({g})}\cdot d(g_{\rev})\}\geq \min\{\abs{\lc(g)},\abs{\tc({g})}\}\cdot \max\{ d(g), d(g_{\rev})\}$, \autoref{cla:maxd} implies that \[\Max  \geq \min\{\abs{\lc(g)},\abs{\tc(g)}\}\cdot \left(\frac{\|g\|_0}{2e}\right)^{\frac{\|g\|_0}{2}-1}\cdot \left(\frac {\deg{g}-\ord{0}{g}}{e}\right)^{\frac{\norm{g}_0}{2}}.\qedhere\]
\end{proof}

 \autoref{thm:intro:norm-bound} immediately follows by combining the estimate from \autoref{lem:max} with \autoref{thm:quotient-coefficient-bound}.
 


    We note that when $g$ does not divide $f$, no such bound holds for the remainder.
    \begin{remark}
        Let $f,g \in \C[x]$
        and let $q, r$
        be the quotient and the remainder
        of $f$ divided by $g$, respectively.
        For any $\norm{f}_{\infty}$, $\norm{g}_{\infty}$, $\deg{f}$, and $\deg{g}$,
        there exists $f$ and $g$ of sparsity $2$ such that
        $\norm{r}_{\infty} = \norm{f}_{\infty}\norm{g}_{\infty}^{\deg{f} - \deg{g} + 1}$.
    \end{remark}
    \begin{proof}
        Let $d_1, d_2 \in \N$ and $a \in \C$ with $\abs{a}\geq 1$.
        We conclude the claim by noting that
        \begin{equation*}
            x^{d_1} \bmod( x^{d_2 + 1} - ax^{d_2}) = a^{d_1 - d_2}x^{d_2}.\qedhere
        \end{equation*}
    \end{proof}
	
	\begin{corollary}
		In the setting above, if either $g\in\C[x]$ is monic or $g \in \Z[x]$, and in both cases $g(0)\neq 0$, then 		
			\begin{equation*}
			\norm{h}_2
			\leq  { \sqrt{\frac{\norm{g}_0}{\deg g}}\norm{f}_1} 
			\left(\frac{2\norm{g}_0 \cdot L\left( 12\norm{g}_0\deg{g} + 2\deg{f} \right)^2}{\sqrt{\left(\frac{\norm{g}_0}{2e}\right)\cdot \left(\frac{\deg{g}}{e}\right)}}\right)^{\norm{g}_0-1}
			.
		\end{equation*}		
	\end{corollary}
	\begin{proof}
				As $\min\{\abs{\lc(g)},\abs{\tc(g)}\}\geq 1$, the claim follows from \autoref{thm:quotient-coefficient-bound} and \autoref{lem:max}.
	\end{proof}

    The proof of \autoref{thm:div-alg} follows from the bound
        proved in \autoref{thm:quotient-coefficient-bound}, \autoref{lem:max}
        and the analysis of the polynomial division algorithm given in \autoref{sec:division-test-in-z}.

\section{Polynomial Division Algorithm}
    \label{sec:division-test-in-z}
    In this section, we analyze the polynomial division algorithm (see  \cite[Chapter 1.5]{cox:2013}
    or \cite[Chapter 2.4]{modern-computer-algebra:2013}),
    also called long division or Euclidean division algorithm.
	We consider a version in which we are given upper bounds on the sparsity and height of the quotient polynomial. 
	Namely, if the sparsity of the quotient polynomial is bounded by $s$,
	and its coefficients are bounded (in absolute value) by $c$, then
	we execute $s$ steps of the division algorithm
	and return True if and only if at termination the remainder is zero.
	If at any step $\norm{q}_{\infty} > c$, we return False.
    
\ignore{
        we obtain a polynomial time algorithm for testing divisibility
        by binomials and cyclotomic-free trinomials over $\Z$
        using \emph{polynomial division algorithm}
        (see  \cite[Chapter 1.5]{cox:2013}
        or \cite[Chapter 2.4]{modern-computer-algebra:2013}),
        also called long division or Euclidean division algorithm.
    Recall that \autoref{thm:intro:binomial-sparsity-bound}
        and \autoref{thm:intro:trinomial-sparsity-bound}
        bound the sparsity of the quotient
        when the division is exact.
    Thus, if during the algorithm the sparsity of the quotient exceeds these bounds,
        we can terminate the algorithm and declare that the division
        is not exact.
    Similarly, we use \autoref{thm:intro:norm-bound}
        to terminate the algorithm when the coefficients
        of the quotient become too large.
    In other words, if the sparsity of the quotient is bounded by $s$,
        and its coefficients are bounded (in absolute value) by $c$, then
        we execute $s$ steps of the division algorithm
        and return True if and only if the remainder is zero.
    If at any step $\norm{h}_{\infty} > c$, we return False.
}

    \begin{algorithm}[H]
        \caption{Bounded Polynomial Division over $\Z$}
        \label{alg:bounded-long-division}
        \begin{algorithmic}
            \Input $g, f \in \Z[x]$,
            a  bound  $s \in \N$ on $\|f/g\|_0$,
            and a bound $c \in \N$ on $\|f/g\|_\infty$.
            \Output True and a factor $f/g$ if $g \mid f$;
            False otherwise.
        \end{algorithmic}
        \begin{enumerate}
            \item Set $q = 0$, $r = f$, $i = 1$, and $t = 0$.
            \item While $\deg{r} \geq \deg{g}$, $|\lc(t)| \leq c$, and $i \leq s$:
            \begin{enumerate}
                \item If $\lc(g) \nmid \lc(r)$, return False.
                \item Set $t = LT(r)/LT(g)$.
                \item Set $q = q + t$.
                \item Set $r = r - t \cdot g$.
                \item Set $i = i + 1$.
            \end{enumerate}
            \item If $r = 0$ return True and the quotient polynomial $q$;
            Otherwise, return False.
        \end{enumerate}
    \end{algorithm}

    The correctness of \autoref{alg:bounded-long-division}
        is due to the fact that in every step of the loop,
        we reveal a single term of the quotient.
    By definition of $c$ and $s$,
        the coefficient of the term cannot exceed $c$,
        and the number of terms cannot be larger than $s$,
        unless $g \nmid f$.

    \begin{proposition}
        \label{prop:bounded-long-division}
        The time complexity of \autoref{alg:bounded-long-division} is
        \begin{equation*}
            \widetilde{O}\left(
                (\norm{f}_0 +  s \cdot \norm{g}_0)
                \cdot (\log{\norm{g}_{\infty}} + \log{\norm{f}_\infty} + \log{c} + \log\deg{f})
            \right)
        \end{equation*}
    \end{proposition}
    \begin{proof}
        Each iteration of the While loop requires $O(\norm{g}_0)$
            multiplications, additions and divisions of integers
            whose absolute values are at most
        \begin{equation*}
            B \leq  \norm{f}_\infty+ c \cdot \norm{g}_\infty
            .
        \end{equation*}
        Hence, its bit complexity is at most $b = \ceil{\log_2 B}$.
        These arithmetic operations can be performed in time 
        \begin{equation*}
            O(b\log b) = \widetilde{O}( \log c + \log \norm{g}_\infty + \log \norm{f}_\infty)
        \end{equation*} (see \cite{harvey2021integer}).
        In addition, we perform operations on the exponent that require $O(\log\deg f)$ time and we have to store both $q$ and $r$ in memory.
        Storing $q$ requires storing $s$ terms,
            each term has degree at most $\deg f$
            and its coefficient is at most $c$ in absolute value.
        Thus, the cost of storing $q$ is at most $s(\log_2{\deg{f}} + \log_2{c})$.
        We next bound the complexity of storing the remainder, $r$.
        In each iteration of the While loop
            the number of terms of $r$
            increases (additively) by at most $\norm{g}_0 - 1$.
        Similarly, $\norm{r}_{\infty}$ increases (additively) by
            at most $\leq \norm{g}_{\infty} \cdot c$ in each step.
        Hence, in every step of the algorithm, it holds that
        \begin{equation*}
            \begin{aligned}
                \norm{r}_0 &\leq \norm{f}_0 + s(\norm{g}_0-1), \\ 
                \norm{r}_{\infty} &
                    \leq {\norm{f}_{\infty}} + s\cdot c \cdot {\norm{g}_{\infty}}.
            \end{aligned}
        \end{equation*}
        This implies that the representation size of $r$ is at most
        \begin{equation*}
            \begin{aligned}
                \size{r} &= \norm{r}_0\cdot(\log_2{\norm{r}_\infty} + \log_2{\deg{r}}) \\
                    &\leq (\norm{f}_0 +  s \cdot (\norm{g}_0 - 1))
                    \cdot ( \log_2{\norm{f}_{\infty}} +\log_2{\norm{g}_{\infty}} + \log_2{s} + \log_2{c} + \log_2{\deg{f}})
                .
            \end{aligned}
        \end{equation*}
        Combining all these estimates we get that the running time is at most 
        \begin{align*}
            &\widetilde{O}\left((\norm{f}_0 +  s \cdot \norm{g}_0)
            \cdot (\log{\norm{g}_{\infty}} + \log{\norm{f}_\infty} + \log{c} + \log\deg{f}) \right).
            \qedhere
        \end{align*}
    \end{proof}

    The proof of \autoref{thm:div-alg} easily follows.

    \begin{proof}[Proof of \autoref{thm:div-alg}]
        Let $s =\deg(f)\geq  \norm{f/g}_0$ and $c =\log_2{\norm{f}_{\infty}} + O(\|g\|_0\cdot \log{\deg{f}})\geq  \norm{f/g}_\infty$ be as in Equation~\eqref{eq:bound-bit-Z}. Note that $\deg f$ is a very rough upper bound on $\|f/g\|_0$, but as the division is exact the algorithm will terminate after $\norm{f/g}_0$ iterations. Thus, we may assume in the upper bound calculations that $s=\norm{f/g}_0$.
        Substituting these values to \autoref{prop:bounded-long-division}
            (and observing that $\norm{f}_0 \leq \norm{g}_0\norm{f/g}_0$),
            gives the desired result. 
    \end{proof}

            \section*{Acknowledgment}
			The authors would like to thank Bruno Grenet
                for his helpful comments and for suggesting \autoref{cor:rand-div}. We would also like to thank 
                the anonymous referees for their comments.

	\bibliographystyle{alpha}
	\bibliography{refs}

\end{document}

%% file: thmmacros.tex
\usepackage{amsthm}
\usepackage{thmtools,thm-restate}

\numberwithin{equation}{section}
\declaretheoremstyle[bodyfont=\it,qed=\qedsymbol]{noproofstyle}

\declaretheorem[numberlike=equation]{observation}

\declaretheorem[name=Observation,numbered=no]{observation*}

\declaretheorem[numberlike=equation]{fact}

\declaretheorem[numberlike=equation]{theorem}

\declaretheorem[name=Theorem,numbered=no]{theorem*}

\declaretheorem[numberlike=equation]{lemma}
\declaretheorem[name=Lemma,numbered=no]{lemma*}

\declaretheorem[numberlike=equation]{corollary}
\declaretheorem[name=Corollary,numbered=no]{corollary*}

\declaretheorem[numberlike=equation]{proposition}
\declaretheorem[name=Proposition,numbered=no]{proposition*}

\declaretheorem[numberlike=equation]{claim}
\declaretheorem[name=Claim,numbered=no]{claim*}

\declaretheorem[name=Conjecture,numbered=no]{conjecture*}

\declaretheorem[name=Question,numbered=no]{question*}

\declaretheorem[name=Open Problem]{openproblem}

\declaretheoremstyle[bodyfont=\it,qed=$\lozenge$]{defstyle} 

\declaretheorem[numberlike=equation,style=defstyle]{definition}
\declaretheorem[unnumbered,name=Definition,style=defstyle]{definition*}

\declaretheorem[unnumbered,name=Example,style=defstyle]{example*}

\declaretheorem[unnumbered,name=Notation=defstyle]{notation*}

\declaretheorem[unnumbered,name=Construction,style=defstyle]{construction*}

\declaretheorem[numberlike=equation,style=defstyle]{remark}
\declaretheorem[unnumbered,name=Remark,style=defstyle]{remark*}


%% file: customurlbst/bibmacros.tex
\usepackage{nth}
\usepackage{intcalc}
\usepackage{etoolbox}
\usepackage{xstring}

\usepackage{ifpdf}
\ifpdf
\else
\usepackage[quadpoints=false]{hypdvips}
\fi

\newcommand{\ECCCurl}[2]{http://eccc.hpi-web.de/report/\ifnumcomp{#1}{>}{93}{19}{20}#1/#2/}

\newcommand{\shortECCC}[2]{\texttt{\href{http://eccc.hpi-web.de/report/\ifnumcomp{#1}{>}{93}{19}{20}#1/#2/}{eccc:TR#1-#2}}}

\newcommand{\parseECCC}[1]{
\StrSubstitute{#1}{TR}{}[\tmpstring]%
\IfSubStr{\tmpstring}{/}{ 
\StrBefore{\tmpstring}{/}[\ecccyear]%
\StrBehind{\tmpstring}{/}[\ecccreport]%
}{
\StrBefore{\tmpstring}{-}[\ecccyear]%
\StrBehind{\tmpstring}{-}[\ecccreport]%
}%
\shortECCC{\ecccyear}{\ecccreport}}